\RequirePackage{fix-cm}

\documentclass[smallextended]{svjour3}       

\usepackage[T1]{fontenc}
\usepackage[utf8]{inputenc}

\let\vec\relax 
\DeclareMathAccent{\vec}{\mathord}{letters}{"7E} 

\usepackage[sumlimits,intlimits]{mathtools}
\usepackage{newtxtext,newtxmath}

\smartqed  
\usepackage{graphicx}
\usepackage{amssymb}
\usepackage{amsfonts}
\usepackage{mathrsfs}
\usepackage{color}
\usepackage{cancel}
\usepackage{graphicx}
\usepackage{epstopdf}
\usepackage{cite,url}
\usepackage{upquote}
\usepackage{pifont}
\usepackage{comment}
\usepackage{appendix}

\usepackage{units}
\newcommand{\I}{\mathscr{I}}
\newcommand{\Ia}{\mathscr{I}_{\alpha}}
\newcommand{\X}{\mathbb{X}}

\makeatletter
\newcommand{\vast}{\bBigg@{4}}
\newcommand{\Vast}{\bBigg@{5}}
\makeatother
\renewcommand{\vec}[1]{\mathbf{#1}} 

\journalname{Information Geometry}

\begin{document}

\title{Cram\'er-Rao Lower Bounds Arising from Generalized Csisz\'ar Divergences
}


\author{M.~Ashok~Kumar 
\and Kumar~Vijay~Mishra
}


\institute{M. A. Kumar \at
              Department of Mathematics, Indian Institute of Technology Palakkad, 678557 India \\
             \email{ashokm@iitpkd.ac.in}
           \and
            K. V. Mishra \at
            United States Army Research Laboratory, Adelphi, MD 20783 USA\\
            \email{kumarvijay-mishra@uiowa.edu} 
}

\date{Received: date / Accepted: date}

\maketitle

\begin{abstract}
We study the geometry of probability distributions with respect to a generalized family of Csisz\'ar $f$-divergences. A member of this family is the relative $\alpha$-entropy which is also a R\'enyi analog of relative entropy in information theory and known as logarithmic or projective power divergence in statistics. We apply Eguchi's theory to derive the Fisher information metric and the dual affine connections arising from these generalized divergence functions. This enables us to arrive at a more widely applicable version of the Cram\'{e}r-Rao inequality, which provides a lower bound for the variance of an estimator for an escort of the underlying parametric probability distribution. We then extend the Amari-Nagaoka's dually flat structure of the exponential and mixer models to other distributions with respect to the aforementioned generalized metric. We show that these formulations lead us to find unbiased and efficient estimators for the escort model. Finally, we compare our work with prior results on generalized Cram\'er-Rao inequalities that were derived from non-information-geometric frameworks.

\end{abstract}

\section{Introduction}
	\label{sec:introduction}
	The relative entropy or Kullback-Leibler divergence between two probability distributions of a random variable is regarded as a measure of ``distance" between them \cite{cover2012elements}. A quantity of fundamental importance in probability, statistics, and information theory, it establishes, \textit{inter alia}, the maximum entropy principle for decision making under uncertainty. In statistics, it is observed as the expected logarithm of the likelihood ratio. In information theory, it arises as the penalty in expected compressed length while using a wrong distribution for compression. 
	The relative entropy of a probability mass function (PMF) $p$ with respect to another PMF $q$ on an alphabet set, say $\X = \{0,1,2,\dots,M\}$ is defined as
	\begin{align}
	    \mathscr{I}_(p,q) := \sum_{x\in\mathbb{X}} p(x)\log\frac{p(x)}{q(x)}.
	\end{align}
	The Shannon entropy is defined as $H(p) := -\sum_{x\in\mathbb{X}} p(x)\log p(x)$. Relative entropy and Shannon entropy are related by $\mathscr{I}(p,u) = \log |\mathbb{X}| - H(p)$, where $u$ is the uniform distribution on $\mathbb{X}$. Throughout the paper, we shall assume that support of all probability distributions is $\mathbb{X}$.
	
	There are other measures of uncertainty that are used as alternatives to Shannon entropy. 
	One of these is the R\'{e}nyi entropy that was discovered by Alfred R\'enyi while attempting to find an axiomatic characterization to measures of uncertainty \cite{renyi1961measures}. Later, Campbell gave an operational meaning to R\'enyi entropy  \cite{1965xxIC_Cam}; he showed that R\'enyi entropy plays the role of Shannon entropy in a source coding problem where normalized cumulants of compressed lengths are considered instead of expected compressed lengths. Blumer and McEliece \cite{198809TIT_BluMcE} and Sundaresan \cite{200701TIT_Sun} studied the mismatched (source distribution) version of this problem and showed that relative $\alpha$-entropy plays the role of relative entropy in this problem. The R\'{e}nyi entropy of $p$ of order $\alpha$, $\alpha \ge 0$, $\alpha\neq 1$, is defined to be $H_{\alpha}(p) := \frac{1}{1-\alpha}\log\sum_xp(x)^{\alpha}$. Let us define {\em the relative $\alpha$-entropy} of $p$ with respect to $q$ as
	\begin{eqnarray}
	\label{eqn:alphadiv}
      \mathscr{I}_{\alpha}(p,q)	& := & \frac{1}{1-\alpha} \log \sum_x p(x) q(x)^{\alpha-1} - \frac{1}{\alpha(1-\alpha)}\log \sum_x p(x)^{\alpha} + \frac{1}{\alpha}\log \sum_x q(x)^{\alpha}.\nonumber\\
	\end{eqnarray}
	 It follows that, as $\alpha \rightarrow 1$, we have $\mathscr{I}_{\alpha}(p,q) \rightarrow \mathscr{I}(p,q)$ 
	 and $H_{\alpha}(p) \rightarrow H(p)$ \cite{kumar2015minimization-1}. R\'enyi entropy and relative $\alpha$-entropy are related by the equation $\mathscr{I}_\alpha(p,u) = \log |\mathbb{X}| - H_\alpha(p)$. 
	Relative $\alpha$-entropy is closely related to the Csisz\'ar $f$-divergence $D_{f}$ as
	\begin{equation}
    \label{eqn:alphadiv2}
    \mathscr{I}_{\alpha}(p,q) = \frac{1}{1-\alpha} \log\left[ \text{sgn}(1-\alpha) \cdot D_{f}(p^{(\alpha)},q^{(\alpha)}) + 1\right],
    \end{equation}
where	
    \begin{equation*}
    \label{eqn:the-function-in_D_f}
    p^{(\alpha)}(x) := \frac{p(x)^{\alpha}}{\sum_y {p(y)}^{\alpha}},\quad q^{(\alpha)}(x) := \frac{q(x)^{\alpha}}{\sum_y {q(y)}^{\alpha}}, \text{ and } f(u) = \text{sgn}(1-\alpha) \cdot (u^{{1}/{\alpha}} - 1), u \geq 0,
    \end{equation*}
	  [c.f. \cite[Sec.~II]{kumar2015minimization-1}]. The measures $p^{(\alpha)}$ and $q^{(\alpha)}$ are called {\em $\alpha$-escort} or {\em $\alpha$-scaled} measures \cite{1998xxPhyA_Tsa}, \cite{201802NCC_KarSun}. Observe from \ref{eqn:the-function-in_D_f} that relative $\alpha$-entropy is a monotone function of the Csisz\'ar divergence, not between $p$ and $q$, but their escorts $p^{(\alpha)}$ and $q^{(\alpha)}$.
	  
	  For a strictly convex function $f$ with $f(1) = 0$, the Csisz\'ar $f$-divergence between two probability distributions $p$ and $q$ is defined as (also, see \cite{1991xxTAS_Csi})
	  \begin{equation*}
	      D_f(p,q) = \sum_x q(x) f\left(\frac{p(x)}{q(x)}\right).
	  \end{equation*}
	  A simple derivation shows, indeed, that the right side of (\ref{eqn:alphadiv2}) is R\'enyi divergence between $p^{(\alpha)}$ and $q^{(\alpha)}$ of order ${1}/{\alpha}$ \cite{kumar2015minimization-2}. For an extensive study of properties of the R\'enyi divergence, we refer the reader to \cite{201407TIT_ErvHar}. Note that, even though the relative $\alpha$-entropy is connected to the Csisz\'ar divergence which, in turn, is linked to the Bregman divergence $B_f$ through
	  \[
	  D_f(p, q) = \sum_i p_i B_f({q_i}/{p_i},1)
	  \]
	  \cite{Zhang2004}, the relative $\alpha$-entropy is quite different from both Csisz\'ar and Bregman divergences because of the appearance of the escort distributions in (\ref{eqn:alphadiv2}).
	 
	 The ubiquity of R\'enyi entropy and relative $\alpha$-entropy in information theory was further noticed, for example, in guessing problems by Ar{\i}kan \cite{199601TIT_Ari}, Sundaresan \cite{200701TIT_Sun}, and Huleihel et al. \cite{Huleihel}; and in encoding of tasks by Bunte and Lapidoth \cite{2014xxarx_BunLap}. Relative $\alpha$-entropy arises in statistics as a generalized likelihood function robust to outliers \cite{2001xxBio_Jon_etal}, \cite{kumar2015minimization-2}. It has been referred variously as $\gamma$-divergence \cite{2008xxJma_FujEgu,2010xxEnt_CicAma,201402NC_NotKomEgu}, projective power divergence \cite{201109Ent_EguKomKat,201002Ent_EguKat}, and logarithmic density power divergence \cite{2011xxSIMDA_BasShiPar}. Throughout this paper, we shall follow the nomenclature of relative $\alpha$-entropy.
	
	Relative $\alpha$-entropy shares many interesting properties with relative entropy (see, e.g. \cite[Sec.~II]{kumar2015minimization-1} for a summary of its properties and relationships to other divergences). For instance, analogous to relative entropy, relative $\alpha$-entropy behaves like squared Euclidean distance and satisfies a Pythagorean property \cite{kumar2015minimization-1,kumar2018information}. The Pythagorean property proved useful in arriving at a computation scheme \cite{kumar2015minimization-2} for a robust estimation procedure \cite{2008xxJma_FujEgu}. This motivates us to explore relative $\alpha$-entropy from the differential geometric perspective. Eguchi \cite{1992xxHMJ_Egu} suggested a method of defining a {\em Riemannian metric}, and a pair of dual {\em affine connections}, on {\em statistical manifolds} from a general divergence function. 
	If we apply Eguchi's method with relative entropy as the divergence function, the resulting statistical manifold is the one with the Riemannian metric specified by the Fisher information matrix. Moreover, the resulting pair of dual affine connections are well-studied as exponential and mixture connections \cite{2000xxMIG_AmaNag}. 
    
The main contributions of this paper are the following.
\begin{enumerate}
    \item \textbf{Structure of statistical manifold for relative $\alpha$-entropy.} We apply Eguchi's theory to relative $\alpha$-entropy and come up with a generalized Fisher information metric and a pair of connections so that these form a dualistic structure on a statistical model. This coincides with the usual Fisher information metric and the exponential and mixer connections when $\alpha = 1$.
       \vspace{0.1cm}
    
    \item \textbf{The generalized CRLB and efficient estimators.} We derive the $\alpha$-version of the Cram\'er-Rao inequality by following the work of Amari and Nagoaka for relative $\alpha$-entropy. This helps us to find unbiased and efficient estimators for the escort distribution from estimators of the original distribution. In particular, if the escort distribution is exponential, we illustrate the procedure to deduce such estimators.
       \vspace{0.1cm}

    \item \textbf{Extension to generalized Csisz\'ar $f$-divergences.} We generalize the results of relative $\alpha$-entropy in 1) and 2) to a general form of Csisz\'ar $f$-divergences. This improves the applicability of these ideas to more general models.
    \vspace{0.1cm}

    \item \textbf{Counterexample of estimator-distribution duality.} The $\alpha$-power-law family was derived by minimizing relative $\alpha$-entropy subject to linear constraints on the underlying distribution. The $\alpha$-exponential family (comprises the generalized Gaussians as subclass in the continuous case) was derived by minimizing R\'enyi divergence subject to linear constraints on the escort of the underlying distribution. Since these divergences and the families are closely related by the mapping $p\mapsto p^{(\alpha)}$ \cite{201802NCC_KarSun}, one would expect these two families be dual to each other with respect to FIM or the $\alpha$-FIM. However, we show that this is not the case. 
   \vspace{0.1cm}
   
   \item \textbf{Connections with other generalized Cram\'er-Rao inequalities.} We interpret and differentiate our result in the context of similar generalized Cram\'er-Rao inequalities derived through non-information-geometric frameworks. For example, Furuichi \cite{furuichi2009} established a generalized CRLB for the so-called $q$-variance (that is, variance with respect to the escort distribution) by defining a $q$-Fisher information using the $q$-logarithmic functions well-known in non-extensive statistical physics literature. Lutwak et al. \cite{200501TIT_LutYanZha} (see also \cite{Lutwak2012}) proposed an extension of Fisher information in an attempt to extend Stam's inequality for the generalized Gaussians. They derived a generalized Cram\'{e}r-Rao inequality for which the generalized Gaussians are the extremal distributions. Along similar lines, Bercher \cite{Bercher2012,bercher2012generalized} derived a two parameter extension of Fisher information and a Cram\'{e}r-Rao inequality. From a statistical standpoint, Jan Naudts \cite{2004xxJIPAM_Jan} (see also \cite{2011xxGT_JanNaudts}), established a generalized CRLB for the variance of an estimator of the escort of the underlying distribution. Unlike these, we derive the generalized Fisher information and the Cramer-Rao inequality from the information geometric perspective of Amari and Nagoaka \cite{2000xxMIG_AmaNag} and Eguchi \cite{1992xxHMJ_Egu}.
    
\end{enumerate}

	The rest of the paper is organized as follows. In Section~\ref{sec:geometry_of_alpha_relative_entropy}, we provide a brief introduction to information geometry, explain Eguchi's theory of obtaining Riemannian metrics and dual affine connections from general divergence functions. 
	In Section~\ref{sec:analogous_cr_inequality}, we derive the $\alpha$-Cram\'er-Rao inequality and exploit it to establish a relation between the $\alpha$-power-law and $\alpha$-exponential families. Then, we apply this framework to generalized Csisz\'ar $f$-divergences in Section~\ref{sec:general_framework}. We discuss other generalizations of Cram\'{e}r-Rao inequality in Section~\ref{sec:disc} and conclude in Section~\ref{sec:summary}.

	\section{Information Geometry of the Relative $\alpha$-Entropy}
	\label{sec:geometry_of_alpha_relative_entropy}
	We now summarize the information geometric concepts associated with a general divergence function. For detailed mathematical definitions, we refer the reader to Amari and Nagoaka \cite{2000xxMIG_AmaNag}. For more intuitive explanations of information geometric notions, one may refer to Amari's recent book \cite{2016xxIGA_Ama}. Ay et al. \cite{2017xxMIG_AyJosLeSch} treats information geometry also from a functional analytic framework. We shall introduce the reader to a certain dualistic structure on a statistical manifold of probability distributions arising from a divergence function. We shall then specialize these ideas to the geometry associated with the relative $\alpha$-entropy.
	
	A {\em statistical manifold} is a parametric family of probability distributions on $\mathbb{X}$ with a ``continuously varying" parameter space. A statistical manifold $S$ is usually represented by $S = \{p_{\theta}: \theta = (\theta_1,\dots,\theta_n)\in \Theta\subset\mathbb{R}^n\}$. Here $\Theta$ is the parameter space. $\theta_1,\dots,\theta_n$ are the coordinates of the point $p$ in $S$ and the mapping $p \mapsto (\theta_1(p),\dots,\theta_n(p)$ that take a point $p$ to its coordinates constitute a {\em coordinate system}. The ``dimension" of the parameter space is the dimension of the manifold. For example, the set of all binomial probability distributions $\{B(k,\theta):\theta\in (0,1)\}$ is a one-dimensional statistical manifold. The {\em tangent space} at a point $p$ on a manifold $S$ (denoted $T_p(S)$) is a linear space that corresponds to the ``local linearization" of the manifold around the point $p$. The elements of a tangent space are called {\em tangent vectors}. For a coordinate system $\theta$, the (standard) basis vectors of a tangent space $T_p$ are denoted by $(\partial_i)_p := \left({\partial}/{\partial\theta_i}\right)_p, i=1,\dots,n$. A {\em (Riemannian) metric} at a point $p$ is an inner product defined between any two tangent vectors at that point. A metric is completely characterized by the matrix whose entries are the inner products between the basic tangent vectors. That is, it is characterized by the matrix
	\[
	G(\theta) = [g_{i,j}(\theta)]_{i,j = 1,\dots,n},
	\]
	where $g_{i,j}(\theta) := \langle \partial_i, \partial_j\rangle $. An {\em affine connection} (denoted $\nabla$) on a manifold is a correspondence between the tangent vectors at a point $p$ to the tangent vectors at a ``nearby" point $p'$ on the manifold. An affine connection is completely specified by specifying the $n^3$ real numbers $(\Gamma_{ij,k})_p, i,j,k=1,\dots,n$ called the {\em connection coefficients} associated with a coordinate system $\theta$.
	
	Let $D$ be a divergence function\footnote{A divergence function is a non-negative function $D$ on $S\times S$ satisfying $D(p,q) \ge 0$ with equality iff $p=q$.} on $S$. Let $D^*$ be another divergence function defined by $D^*(p,q) = D(q,p)$. Eguchi \cite{1992xxHMJ_Egu} showed that given an $n$-dimensional manifold $S = \{p_{\theta}\}$, with coordinate system $\theta = (\theta_1,\dots,\theta_n)$, and a (sufficiently smooth) divergence function $D$ on $S$, there is a  metric
	\[
	G^{(D)}(\theta) = \left[g_{i,j}^{(D)}(\theta)\right]
	\]
	with
	\begin{eqnarray*}
		g_{i,j}^{(D)}(\theta) & := & -D[\partial_i,\partial_j]\\
		& := & -\frac{\partial}{\partial\theta_i}\frac{\partial}{\partial\theta_j'}D(p_{\theta},p_{\theta'})\bigg|_{\theta=\theta'}
	\end{eqnarray*}
	where $\theta = (\theta_1,\dots,\theta_n)$, $\theta' = (\theta_1',\dots,\theta_n')$, and there are affine connections $\nabla^{(D)}$ and $\nabla^{(D^*)}$, with connection coefficients
	\begin{eqnarray*}
		\Gamma_{ij,k}^{(D)}(\theta) & := & -D[\partial_i\partial_j,\partial_k]\\
		& := & -\frac{\partial}{\partial\theta_i}\frac{\partial}{\partial\theta_j}\frac{\partial}{\partial\theta_k'}D(p_{\theta},p_{\theta'})\bigg|_{\theta=\theta'}
	\end{eqnarray*}
	and
	\begin{eqnarray*}
		\Gamma_{ij,k}^{(D^*)}(\theta) & := & -D[\partial_k,\partial_i\partial_j]\\
		& := & -\frac{\partial}{\partial\theta_k}\frac{\partial}{\partial\theta_i'}\frac{\partial}{\partial\theta_j'}D(p_{\theta},p_{\theta'})\bigg|_{\theta=\theta'},
	\end{eqnarray*}
	such that $\nabla^{(D)}$ and $\nabla^{(D^*)}$ are duals of each other with respect to the metric $G^{(D)}$ in the sense that
	\begin{eqnarray}
	\label{dualistic-structure}
	\partial_k g_{i,j}^{(D)}=\Gamma_{ki,j}^{(D)}+ \Gamma_{kj,i}^{(D^*)}.
	\end{eqnarray}
	When $D(p,q) = I(p \| q)$, the relative entropy, the resulting metric is the {\em Fisher metric} or the {\em  information metric} defined by the {\em Fisher information matrix} $G(\theta) = [g_{i,j}(\theta)]$ with
	\begin{eqnarray}
	\nonumber
	g_{i,j}(\theta)
	& = & \left. - \frac{\partial}{\partial \theta_i} \frac{\partial}{\partial \theta'_j} \sum_{x} p_{\theta}(x) \log \frac{p_{\theta}(x)}{p_{\theta'}(x)} \right|_{\theta' = \theta}\\
	\nonumber
	& = & \sum_x \partial_i p_{\theta}(x)\cdot \partial_j \log p _{\theta}(x)\\
	\nonumber
	& = & E_{\theta}[\partial_i \log p_{\theta}(X)\cdot \partial_j \log p_{\theta}(X)]\\
	\label{eqn:fisher-information-metric}
	& = & \text{Cov}_{\theta}[\partial_i \log p_{\theta}(X), \partial_j \log p_{\theta}(X)].
	\end{eqnarray}
	The last equality follows from the fact that the expectation of the score function is zero, that is, $E_{\theta}[\partial_i \log p_{\theta}(X)] = 0, i = 1, \dots, n$. The affine connection $\nabla^{(I)}$ is also called the {\em mixture connection} and is sometimes denoted $\nabla^{(m)}$ ($m$-connection). The affine connection $\nabla^{(I^*)}$ is also called the {\em exponential connection} and is sometimes denoted $\nabla^{(e)}$  ($e$-connection). The connection coefficients are 
	\begin{eqnarray*}
		\Gamma_{ij,k}^{(m)}(\theta)
		& = & \sum_x \partial_i \partial_j p_{\theta}(x)\cdot \partial_k \log p _{\theta}(x)
	\end{eqnarray*}
	for the $m$-connection and
	\begin{eqnarray*}
		\Gamma_{ij,k}^{(e)}(\theta)
		& = & \sum_x \partial_k p _{\theta}(x)\cdot \partial_i \partial_j \log p _{\theta}(x)
	\end{eqnarray*}
	for the $e$-connection (c.f., \cite[Sec. 3.2]{2000xxMIG_AmaNag}).
	
	Let us see what the Eguchi framework yields when we set $D = \Ia$. For simplicity, write $G^{(\alpha)}$ for $G^{(\Ia)}$. The Riemannian metric on $S$ is specified by the matrix $G^{(\alpha)}(\theta) = [g_{i,j}^{(\alpha)}(\theta)]$, where
	\begin{eqnarray}
	\lefteqn{g_{i,j}^{(\alpha)}(\theta) ~ := ~ g_{i,j}^{(\Ia)} } \nonumber\\
	& = & -\frac{\partial}{\partial\theta_j'}\frac{\partial}{\partial\theta_i}\mathscr{I}_{\alpha}(p_{\theta},p_{\theta'})\bigg|_{\theta' = \theta} \nonumber\\\label{eqn:alpha-metric}\\
    & = & \frac{1}{\alpha-1}\cdot\frac{\partial}{\partial\theta_j'}\frac{\partial}{\partial\theta_i} \left[\log \sum_y p_{\theta}(x) {p_{\theta'}(x)}^{\alpha-1}\right]_{\theta' = \theta}\\
	& = & \frac{1}{\alpha-1}\sum_x \partial_i p_{\theta}(x)\cdot \partial_j'\left[\frac{{p_{\theta'}(x)}^{\alpha-1}}{\sum_y p_{\theta}(y) {p_{\theta'}(y)}^{\alpha-1}}\right]_{\theta' = \theta}\\
    & = & \sum_x \partial_i p_{\theta}(x)\left[\frac{{p_{\theta}(x)}^{\alpha-2}\partial_j p_{\theta}(x) \sum_y p_{\theta}(y)^{\alpha} - p_{\theta}(x)^{\alpha-1}\sum_y p_{\theta}(y)^{\alpha-1}\partial_j p_\theta(y)}{(\sum_y p_{\theta}(y)^{\alpha})^2}\right]\\
	& = & E_{\theta^{(\alpha)}}[\partial_i (\log p_{\theta}(X))\cdot \partial_j (\log p_{\theta}(X))]\nonumber\\
	\label{eqn:g-alpha-expansion}
	& & \hspace{1.5cm}-E_{\theta^{(\alpha)}}[\partial_i \log p_{\theta}(X)]\cdot E_{\theta^{(\alpha)}}[\partial_j \log p_{\theta}(X)]\\\label{eqn:alpha-metric-covariance}
	& = & \text{Cov}_{\theta^{(\alpha)}}[\partial_i \log p_{\theta}(X), \partial_j \log p_{\theta}(X)]\\
	\label{eqn:RiemannianOnS-alpha}
	& = & \frac{1}{\alpha^2}\text{Cov}_{\theta^{(\alpha)}}[\partial_i \log p_{\theta}^{(\alpha)}(X), \partial_j \log p_{\theta}^{(\alpha)}(X)],
	\end{eqnarray}
	where $p_{\theta}^{(\alpha)}$ is the $\alpha$-escort distribution associated with $p_{\theta}$,
	\begin{equation}
	\label{eqn:escort_distribution}
	p_{\theta}^{(\alpha)}(x) := \frac{p_{\theta}(x)^{\alpha}}{\sum_y {p_{\theta}(y)}^{\alpha}},
	\end{equation}
	and $E_{\theta^{(\alpha)}}$ denotes expectation with respect to $p_{\theta}^{(\alpha)}$. The equality (\ref{eqn:RiemannianOnS-alpha}) follows because
	\begin{align*}
	\partial_i p_\theta^{(\alpha)}(x) = \partial_i\left(\frac{p_\theta(x)^\alpha}{\sum_y p_\theta(y)^\alpha}\right) = \alpha\left[\frac{{p_{\theta}^{(\alpha)}(x)}}{p_{\theta}(x)}\partial_i p_{\theta}(x) - p_{\theta}^{(\alpha)}(x) \sum_y \frac{{p_{\theta}^{(\alpha)}(y)}}{p_{\theta}(y)}\partial_i p_{\theta}(y)\right].
	\end{align*}
	
	\begin{remark}
		\label{rem:alpha_fisher}
		If we define $S^{(\alpha)} := \{p_\theta^{(\alpha)} : p_\theta\in S\}$, then (\ref{eqn:RiemannianOnS-alpha}) tells us that $G^{(\alpha)}$ is essentially the usual Fisher information for the model $S^{(\alpha)}$ up to the scale factor $\alpha$.
	\end{remark}
	 We shall call the metric defined by $G^{(\alpha)}$ an \emph{$\alpha$-information metric}. We shall assume that $G^{(\alpha)}$ is positive definite. One example case when this holds is the following.
	
	Consider the entire probability simplex $\mathcal{P}$ on $\X$. It can be parametrized by $\mathcal{P} = \{p_{\theta}\}$, $\theta = (\theta_1,\dots,\theta_n)\in \Theta,$  where
	\begin{eqnarray*}
		\Theta = \left\{(\theta_1,\dots,\theta_n) : \theta_i>0, \sum\limits_{i=1}^n\theta_i < 1\right\},
	\end{eqnarray*}
	with
	\begin{eqnarray}
	\label{parameterization_of_probability}
	p_{\theta}(x) =
	\begin{cases} \theta_x & \text{for } x = 1,\dots,n
	\\
	1-\sum\limits_{i=1}^n\theta_i &\text{for } x = 0.
	\end{cases}
	\end{eqnarray}
	This particular parametrization has been studied by Amari and Nagoaka \cite[Ex.~2.4]{2000xxMIG_AmaNag}.
For any row vector $c = (c_1,\dots,c_n)\in \mathbb{R}^n$, we have
	\begin{eqnarray*}
		cG^{(\alpha)}(\theta)c^t & = & \sum\limits_{i,j}c_i c_j g_{i,j}^{(\alpha)}(\theta)\\
			& = & \sum\limits_{i,j}c_i c_j E_{\theta^{(\alpha)}}[(\partial_i \log p_{\theta}-a_i)\cdot (\partial_j\log p_{\theta}-a_j)]\\
			& = & E_{\theta^{(\alpha)}}\left[\sum\limits_{i,j}c_i c_j (\partial_i \log p_{\theta}-a_i)\cdot (\partial_j\log p_{\theta}-a_j)\right]\\
			& = & E_{\theta^{(\alpha)}}\left[\left(\sum\limits_{i=1}^n c_i (\partial_i \log p_{\theta}-a_i)\right)^2\right],
	\end{eqnarray*}
		where $a_i = E_{\theta^{(\alpha)}}[\partial_i \log p_{\theta}]$. The right-hand side is always non-negative and is $0$ if and only if
		\begin{eqnarray}
		\label{an_equation}
		\sum\limits_{i=1}^n c_i \partial_i \log p_{\theta}(x) = b
		\end{eqnarray}
		for all $x\in \mathbb{X}$, where $b = \sum\limits_{i=1}^n c_i a_i$. Now, for any $x\in \mathbb{X}$,
		\begin{equation*}
		\partial_i \log p_{\theta}(x) =
		\begin{cases}
		-1/\left(1-\sum\limits_{i=1}^n\theta_i\right) & \text{if } x = 0\\
		1/\theta_i & \text{if } i = x\neq 0\\
		0 & \text{otherwise.}
		\end{cases}
		\end{equation*}
		Hence (\ref{an_equation}) holds if and only if
		$$\sum\limits_{i=1}^n c_i =-b \left(1-\sum\limits_{i=1}^n\theta_i\right), \quad c_1 = b\theta_1,\dots,~c_n = b\theta_n$$
		which is possible if and only if $c = 0$. Hence, the metric is positive definite with respect to this particular parameterization of the probability simplex.
\vspace{0.2cm}

	Let us now return to the general manifold $S$ with a coordinate system $\theta$. Denote $\nabla^{(\alpha)}:=\nabla^{(\I_{\alpha})}$ and $\nabla^{(\alpha)*}:=\nabla^{(\I_{\alpha}^*)}$ where the right-hand sides are as defined by Eguchi \cite{1992xxHMJ_Egu} with $D = \I$.
	
	Motivated by the expression for the Riemannian metric in (\ref{eqn:alpha-metric}), define
	\begin{equation}
	\label{alpha-partial}
	\partial_i^{(\alpha)}(p_\theta (x)) := \frac{1}{\alpha-1}\partial_i'\left(\frac{{p_{\theta'}(x)}^{\alpha-1}}{\sum_y p_{\theta}(y)\, {p_{\theta'}(y)}^{\alpha-1}}\right)\bigg |_{\theta'=\theta}.
	\end{equation}
	We now identify the corresponding connection coefficients as
	\begin{eqnarray}
	\label{eqn:connection_coefficients}
	\Gamma_{ij,k}^{(\alpha)}& := & \Gamma_{ij,k}^{(\I_{\alpha})}\\
	& = & - I_{\alpha}[\partial_i\partial_j,\partial_k] \nonumber\\
	& = & \frac{1}{\alpha-1} \left[\sum_x \partial_j p_{\theta}(x) \cdot \partial_i\left(\partial_k^{(\alpha)}(p_\theta)\right) + \sum_x \partial_i\partial_j p_{\theta}(x) \cdot \partial_k^{(\alpha)}(p_\theta)\right]
	\end{eqnarray}
	and
	\begin{eqnarray}
	\label{eqn:dual_connection_coefficients}
	\Gamma_{ij,k}^{(\alpha)*} & := & \Gamma_{ij,k}^{(\I_{\alpha}^*)}\\ & = & -I_{\alpha}[\partial_k,\partial_i\partial_j] \nonumber\\
	& = &\frac{1}{\alpha-1}\left[\sum_x \partial_k p_{\theta}(x) \cdot \partial_i'\partial_j'\left(\frac{{p_{\theta'}(x)}^{\alpha-1}}{\sum_y p_{\theta}(y){p_{\theta'}(y)}^{\alpha-1}}\right)\bigg |_{\theta'=\theta}\right].\nonumber\\
	\end{eqnarray}
	We also have (\ref{dualistic-structure}) specialized to our setting:
	\begin{eqnarray}
	\label{eqn:dual_connection}
	\partial_k g_{i,j}^{(\alpha)}=\Gamma_{ki,j}^{(\alpha)}+ \Gamma_{kj,i}^{(\alpha)*}.
	\end{eqnarray}
	$(G^{(\alpha)},\nabla^{(\alpha)},\nabla^{(\alpha)*})$ forms a dualistic structure on $S$. We shall call the connection $\nabla^{(\alpha)}$ with the connection coefficients $\Gamma_{ij,k}^{(\alpha)}$, an \emph{$\alpha$-connection}.
	
	Some remarks are in order. When $\alpha = 1$, the metric $G^{(\alpha)}(\theta)$ coincides with the usual Fisher metric and the connections $\nabla^{(\alpha)}$ and $\nabla^{(\alpha)*}$ coincide with the $m$-connection $\nabla^{(m)}$ and the $e$-connection $\nabla^{(e)}$, respectively.
	
	A comparison of the expressions in (\ref{eqn:fisher-information-metric}) and (\ref{eqn:RiemannianOnS-alpha}) suggests that 
	the manifold $S$ with the $\alpha$-information metric may be equivalent to the Riemannian metric specified by the Fisher information matrix on the manifold $S^{(\alpha)} := \{ p_{\theta}^{(\alpha)} : \theta \in \Theta \subset \mathbb{R}^n \}$. This is true to some extent because the Riemannian metric on $S^{(\alpha)}$ specified by the Fisher information matrix is simply $G^{(\alpha)}(\theta) = [g_{ij}^{(\alpha)}(\theta)]$. However, our calculations indicate that the $\alpha$-connection and its dual on $S$ are not the same as the $e$- and the $m$-connections on $S^{(\alpha)}$ except when $\alpha = 1$. The $\alpha$-connection and its dual should therefore be thought of as a parametric generalization of the $e$- and $m$-connections. In addition, the $\alpha$-connections in (\ref{eqn:connection_coefficients}) and (\ref{eqn:dual_connection_coefficients}) are different from the $\alpha$-connection of Amari and Nagaoka \cite{2000xxMIG_AmaNag}, which is a convex combination of the $e$- and $m$-connections.
	
	\section{An $\alpha$-Version of Cram\'{e}r-Rao Inequality}
	\label{sec:analogous_cr_inequality}
	We investigate the geometry of $\mathcal{P}$ with respect to the metric $G^{(\alpha)}$ and the affine connection $\nabla^{(\alpha)}$. Later, we formulate 
	an $\alpha$-equivalent version of the Cram\'{e}r-Rao inequality associated with a submanifold $S$.
	
	Note that $\mathcal{P}$ is an open subset of the affine subspace $\mathcal{A}_1 := \{A\in \mathbb{R}^{\mathbb{X}}:\sum\limits_x A(x) = 1\}$ and the tangent space at each $p \in \mathcal{P}$, $T_p(\mathcal{P})$ is the linear space
	\[
	\mathcal{A}_0 := \{A\in \mathbb{R}^{\mathbb{X}}:\sum\limits_x A(x) = 0\}.
	\]
	For every tangent vector $X\in T_p(\mathcal{P})$, let $X_p^{(e)}(x) := X(x)/p(x)$ at $p$ and call it the \emph{exponential representation of $X$ at $p$}. The collection of exponential representations is then
	\begin{eqnarray}
	\label{exponential_tangent_space}
	T_p^{(e)}(\mathcal{P}) := \{X_p^{(e)}:X\in T_p(\mathcal{P})\} = \{A\in \mathbb{R}^{\mathbb{X}}:E_p[A]=0\},\nonumber\hspace*{-1cm}\\
	\end{eqnarray}
	where the last equality is easy to check. Observe that (\ref{alpha-partial}) is 
	\begin{eqnarray}
	\label{eqn:alpha_representation}
	\partial_i^{(\alpha)}(p_\theta (x))	& = & \frac{1}{\alpha-1}\partial_i'\left(\frac{{p_{\theta'}(x)}^{\alpha-1}}{\sum_y p_{\theta}(y) {p_{\theta'}(y)}^{\alpha-1}}\right)\bigg |_{\theta'=\theta}\nonumber\\
	& = & \left[\frac{{p_{\theta}(x)}^{\alpha-2}~\partial_i p_{\theta}(x)}{\sum_y {p_{\theta}(y)}^{\alpha}} - \frac{{p_{\theta}(x)}^{\alpha-1}~\sum_y {p_{\theta}(y)}^{\alpha-1}\partial_i p_{\theta}(y)}{(\sum_y {p_{\theta}(y)}^{\alpha})^2}\right]\nonumber\\
	& = & \left[\frac{{p_{\theta}(x)}^{(\alpha)}}{p_{\theta}(x)}\partial_i(\log p_{\theta}(x)) - \frac{{p_{\theta}(x)}^{(\alpha)}}{p_{\theta}(x)}E_{\theta^{(\alpha)}}[\partial_i(\log p_{\theta}(X))]\right].
	\end{eqnarray}
	We shall call the above an \emph{$\alpha$-representation of $\partial_i$ at $p_\theta$}. With this notation, the $\alpha$-information metric is 
	\begin{equation*}
	g_{i,j}^{(\alpha)}(\theta) = \sum_x \partial_i p_{\theta}(x) \cdot \partial_j^{(\alpha)}(p_\theta(x)).
	\end{equation*}
	It should be noted that $E_{\theta}[\partial_i^{(\alpha)}(p_\theta(X))] = 0$. This follows since 
	\[
	\partial_i^{(\alpha)} (p_{\theta}) = \frac{p_{\theta}^{(\alpha)}}{p_{\theta}} \partial_i \log p_{\theta}^{(\alpha)}.
	\]
	When $\alpha =1$, the right hand side of (\ref{eqn:alpha_representation}) reduces to $\partial_i(\log p_{\theta})$.
	
	Motivated by (\ref{eqn:alpha_representation}), the \emph{$\alpha$-representation of a tangent vector $X$ at $p$} is 
	\begin{eqnarray}
	\label{eqn:alpha_rep_tgt_vec}
	X_p^{(\alpha)}(x)
	& := & \left[\frac{p^{(\alpha)}(x)}{p(x)}X_p^{(e)}(x) - \frac{p^{(\alpha)}(x)}{p(x)}E_{p^{(\alpha)}}[X_p^{(e)}]\right]\nonumber\\
	& = & \left[\frac{p^{(\alpha)}(x)}{p(x)}\left(X_p^{(e)}(x) - E_{p^{(\alpha)}}[X_p^{(e)}]\right)\right].
	\end{eqnarray}
	The collection of all such $\alpha$-representations is
	\begin{eqnarray}
	T_p^{(\alpha)}(\mathcal{P}) := \{X_p^{(\alpha)} : X\in T_p(\mathcal{P})\}.
	\end{eqnarray}
	Clearly $E_p[X_p^{(\alpha)}] = 0$. Also, since any $A\in \mathbb{R}^{\mathbb{X}}$ with $E_p[A]=0$ is 
	\begin{eqnarray*}
		A = \left[\frac{p^{(\alpha)}}{p}\left(B-E_{p^{(\alpha)}}[B]\right)\right]
	\end{eqnarray*}
	with $B = \tilde{B}-E_p[\tilde{B}],$ where
	\[
	\tilde{B}(x) := \left[\frac{p(x)}{p^{(\alpha)}(x)} A(x)\right].
	\]
	In view of (\ref{exponential_tangent_space}), we have
	\begin{eqnarray}
	\label{e_space_equalto_alpha_space}
	T_p^{(e)}(\mathcal{P}) = T_p^{(\alpha)}(\mathcal{P}).
	\end{eqnarray}
	Now the inner product between any two tangent vectors $X,Y\in T_p(\mathcal{P})$ defined by the $\alpha$-information metric in (\ref{eqn:alpha-metric}) is 
	\begin{eqnarray}
	\label{eqn:alpha_metric_general}
	\langle X,Y\rangle^{(\alpha)}_p := E_p[X^{(e)}Y^{(\alpha)}].
	\end{eqnarray}
	Consider now an $n$-dimensional statistical manifold $S$, a submanifold of $\mathcal{P}$, together with the metric $G^{(\alpha)}$ as in (\ref{eqn:alpha_metric_general}). Let $T_p^*(S)$ be the dual space (cotangent space) of the tangent space $T_p(S)$ and let us consider for each $Y\in T_p(S)$, the element $\omega_Y\in T_p^*(S)$ which maps $X$ to $\langle X,Y\rangle^{(\alpha)}$.  The correspondence $Y\mapsto \omega_Y$ is a linear map between $T_p(S)$ and $T_p^*(S)$. An inner product and a norm on $T_p^*(S)$ are naturally inherited from $T_p(S)$ by
	\[
	\langle \omega_X,\omega_Y\rangle_p := \langle X,Y\rangle^{(\alpha)}_p
	\]
	and
	\[
	\|\omega_X\|_p := \|X\|_p^{(\alpha)} = \sqrt{\langle X,X\rangle^{(\alpha)}_p}.
	\]
	Now, for a (smooth) real function $f$ on  $S$, the \emph{differential} of $f$ at $p$, $(\text{d}f)_p$, is a member of $T_p^*(S)$ which maps $X$ to $X(f)$. The \emph{gradient of $f$ at p} is the tangent vector corresponding to $(\text{d}f)_p$, hence, satisfies
	\begin{eqnarray}
	\label{eqn:differential_of_function}
	(\text{d}f)_p(X) = X(f) = \langle (\text{grad} f)_p,X\rangle_p^{(\alpha)},
	\end{eqnarray}
	and
	\begin{eqnarray}
	\label{eqn:norm_of_differential}
	\|(\text{d}f)_p\|_p^2 = \langle (\text{grad}f)_p,(\text{grad}f)_p\rangle_p^{(\alpha)}.
	\end{eqnarray}
	Since $\text{grad}f$ is a tangent vector, 
	\begin{equation}
	\label{eqn:grad-f}
	\text{grad}f = \sum\limits_{i=1}^n h_i \partial_i
	\end{equation}
	for some scalars $h_i$. Applying (\ref{eqn:differential_of_function}) with $X = \partial_j$, for each $j=1,\dots,n$, and using (\ref{eqn:grad-f}), we obtain
	\begin{eqnarray*}
		(\partial_j)(f)
		& = & \left\langle \sum\limits_{i=1}^n h_i \partial_i, \partial_j\right\rangle^{(\alpha)}\\
		& = & \sum\limits_{i=1}^n h_i \langle \partial_i, \partial_j\rangle^{(\alpha)}\\
		& = & \sum\limits_{i=1}^n h_i g_{i,j}^{(\alpha)}, \quad j = 1, \dots, n.
	\end{eqnarray*}
	This yields
	\[
	[h_1,\dots,h_n]^T = \left[G^{(\alpha)}\right]^{-1}[\partial_1(f),\dots,\partial_n(f)]^T,
	\]
	and so
	\begin{equation}
	\label{eqn:grad-coeff-equation}
	\text{grad}f = \sum\limits_{i,j} (g^{i,j})^{(\alpha)}\partial_j(f) \partial_i.
	\end{equation}
	From (\ref{eqn:differential_of_function}), (\ref{eqn:norm_of_differential}), and (\ref{eqn:grad-coeff-equation}), we get
	\begin{eqnarray}
	\label{differential_and_metric}
	\|(\text{d}f)_p\|_p^2 = \sum\limits_{i,j} (g^{i,j})^{(\alpha)}\partial_j(f) \partial_i(f)
	\end{eqnarray}
	where $(g^{i,j})^{(\alpha)}$ is the $(i,j)$th entry of the inverse of $G^{(\alpha)}$.
	
	With these preliminaries, we 
	now state our main results. These are analogous to those in \cite[Sec.~2.5]{2000xxMIG_AmaNag}.
	\begin{theorem}
	\label{thm:variance_and_norm_of_differential}
		Let $A:\mathbb{X}\to\mathbb{R}$ be any mapping (that is, a vector in $\mathbb{R}^{\mathbb{X}}$. Let $E[A]:\mathcal{P}\to \mathbb{R}$ be the mapping $p\mapsto E_p[A]$. We then have
		\begin{eqnarray}
		\label{eqn:variance_and_norm_of_differential}
		\text{Var}_{p^{(\alpha)}}\left[\frac{p}{p^{(\alpha)}}(A-E_p[A])\right] =  \|(\text{d}E_p[A])_p\|_p^2,
		\end{eqnarray}
		where the subscript $p^{(\alpha)}$ in Var means variance with respect to $p^{(\alpha)}$.
		$\hfill$
	\end{theorem}
	\begin{proof}
		For any tangent vector $X\in T_p(\mathcal{P})$,
		\begin{eqnarray}
		\label{eqn:tangent_acting_on_expectation}
		X(E_p[A])
		& = & \sum\limits_x X(x)A(x)\nonumber\\
		& = & E_p[X_p^{(e)} \cdot A]\\
		& = & E_p[X_p^{(e)}(A-E_p[A])].
		\end{eqnarray}
		Since $A-E_p[A]\in T_p^{(\alpha)}(\mathcal{P})$ (c.f.~(\ref{e_space_equalto_alpha_space})), there exists $Y\in T_p(\mathcal{P})$ such that $A-E_p[A] = Y_p^{(\alpha)}$, and $\text{grad}(E[A]) = Y$. Hence we see that
		\begin{eqnarray*}
		\|(\text{d}E[A])_p\|_p^2 & = & E_p[Y_p^{(e)}Y_p^{(\alpha)}]\\
			& = &  E_p[Y_p^{(e)}(A-E_p[A])]\\
			& \stackrel{(a)}{=} &\displaystyle E_p\left[\left\{\frac{p(X)}{ p^{(\alpha)}(X)} (A-E_p[A]) + E_{p^{(\alpha)}}[Y_p^{(e)}]\right\}(A-E_p[A])\right]\\
			& \stackrel{(b)}{=} & E_p\left[\frac{p(X)}{p^{(\alpha)}(X)}(A-E_p[A])(A-E_p[A])\right]\\
			& = &  E_{p^{(\alpha)}}\left[\frac{p(X)}{p^{(\alpha)}(X)}(A-E_p[A])\frac{p(X)}{p^{(\alpha)}(X)}(A-E_p[A])\right]\\
			& = &  \text{Var}_{p^{(\alpha)}}\left[\frac{p(X)}{p^{(\alpha)}(X)}(A-E_p[A])\right],
		\end{eqnarray*}
		where the equality (a) is obtained by applying (\ref{eqn:alpha_rep_tgt_vec}) to $Y$ and (b) follows because $E_p[A-E_p[A]] = 0$.
		\begin{flushright}$\blacksquare$\end{flushright}
	\end{proof}
	
	\begin{corollary}
		\label{cor:variance_and_norm_of_differential_inequality}
		If $S$ is a submanifold of $\mathcal{P}$, then
		\begin{eqnarray}
		\label{variance_and_norm_of_differential}
		\text{Var}_{p^{(\alpha)}}\left[\frac{p(X)}{p^{(\alpha)}(X)}(A-E_p[A])\right] \ge \|(\text{d}E[A]|_{S})_p\|_p^2
		\end{eqnarray}
		with equality if and only if $$A-E_p[A]\in \{X_p^{(\alpha)} : X\in T_p(S)\} =: T_p^{(\alpha)}(S).$$ $\hfill$
	\end{corollary}
	\begin{proof}
		Since $(\text{grad }E[A]|_{S})_p$ is the orthogonal projection of $(\text{grad }E[A])_p$ onto $T_p(S)$, the proof follows from Theorem~\ref{thm:variance_and_norm_of_differential}.\begin{flushright}$\blacksquare$\end{flushright}
	\end{proof}
	
	We use the aforementioned ideas to establish an $\alpha$-version of the Cram\'er-Rao inequality for the $\alpha$-escort of the underlying distribution. This gives a lower bound for the variance of the unbiased estimator $\hat{\theta}^{(\alpha)}$ in $S^{(\alpha)}$.
	
	\begin{theorem} [$\alpha$-version of Cram\'{e}r-Rao inequality]
	\label{thm:alpha_CRLB}
	Let $S = \{p_{\theta} : \theta = (\theta_1,\dots,\theta_m)\in\Theta\}$ be the given statistical model. Let $\hat{\theta}^{(\alpha)} = (\hat{\theta}^{(\alpha)}_1,\dots,\hat{\theta}^{(\alpha)}_m)$ be an unbiased estimator of $\theta = (\theta_1,\dots,\theta_m)$ for the statistical model $S^{(\alpha)} := \{p_\theta^{(\alpha)} : p_\theta\in S\}$. Then $\text{Var}_{\theta^{(\alpha)}}[\hat{\theta}^{(\alpha)}(X)] - [G^{(\alpha)}]^{-1}$ is positive semi-definite, where $\theta^{(\alpha)}$ denotes expectation with respect to $p_{\theta}^{(\alpha)}$. On the other hand, given an unbiased estimator $\hat{\theta} = (\hat{\theta}_1,\dots,\hat{\theta}_m)$ of $\theta$ for $S$, there exists an unbiased estimator $\hat{\theta}^{(\alpha)} = (\hat{\theta}^{(\alpha)}_1,\dots,\hat{\theta}^{(\alpha)}_m)$ of $\theta$ for $S^{(\alpha)}$ such that $\text{Var}_{\theta^{(\alpha)}}[\hat{\theta}^{(\alpha)}(X)] - [G^{(\alpha)}]^{-1}$ is positive semi-definite..
	\end{theorem}
	
	\begin{proof}
		Given an unbiased estimator $\hat{\theta}^{(\alpha)} = (\hat{\theta}^{(\alpha)}_1,\dots,\hat{\theta}^{(\alpha)}_m)$ of $\theta = (\theta_1,\dots,\theta_m)$ for the statistical model $S^{(\alpha)}$, let
	\begin{eqnarray}
	\label{unbiased_estimator_for_S}
	\hat{\theta_i}(X) := \frac{p_{\theta}^{(\alpha)}(X)}{p_{\theta}(X)}\hat{\theta}_i^{(\alpha)}(X).
	\end{eqnarray}
	It is easy to check that $\hat{\theta}$ is an unbiased estimator of $\theta$ for $S$. Hence, if we let $A = \sum\limits_{i=1}^m c_i \hat{\theta_i}$, for $c = (c_1,\dots,c_m)\in \mathbb{R}^m$, then from (\ref{variance_and_norm_of_differential}) and (\ref{differential_and_metric}), we have
	\begin{eqnarray}
	\label{eq:analogous_cramer_rao_inequality}
	c\text{Var}_{\theta^{(\alpha)}}[\hat{\theta}^{(\alpha)}(X)]c^t \ge c[G^{(\alpha)}]^{-1}c^t.
	\end{eqnarray}
	This proves the first part.
	
	For the converse, consider an unbiased estimator $\hat{\theta} = (\hat{\theta}_1,\dots,\hat{\theta}_m)$ of $\theta$ for $S$. Let
	\begin{eqnarray}
	\label{unbiased_estimator_alpha_version}
	\hat{\theta}_i^{(\alpha)}(X) := \frac{p_{\theta}(X)}{p_{\theta}^{(\alpha)}(X)}\hat{\theta_i}(X).
	\end{eqnarray}
	This is an unbiased estimator of $\theta_i$ for $S^{(\alpha)}$.  Hence, the assertion follows from the first part of the proof. \begin{flushright}$\blacksquare$\end{flushright}
	\end{proof}
	
	When $\alpha = 1$, the inequality in (\ref{eq:analogous_cramer_rao_inequality}) reduces to the classical Cram\'{e}r-Rao inequality. In view of Remark \ref{rem:alpha_fisher} in section \ref{sec:geometry_of_alpha_relative_entropy}, we see that (\ref{eq:analogous_cramer_rao_inequality}) is, in fact, the Cram\'{e}r-Rao inequality for the $\alpha$-escort family $S^{(\alpha)}$.

We now introduce the two families of probability distributions that are widely known in the context of Tsallis' thermostatistics.

	\begin{definition}
		\label{defn:alpha-power-law}
		The {\em $\alpha$-power-law family}, $\mathbb{M}^{(\alpha)} := \mathbb{M}^{(\alpha)}(q, f, \Theta)$
		, characterized by a $q\in\mathcal{P}$ and $k$ real valued functions $f_i, i = 1,\dots, k$ on $\mathcal{X}$, and parameter space $\Theta\subset\mathbb{R}^k$, is defined by	$\mathbb{M}^{(\alpha)} = \{ p_\theta : \theta\in\Theta\}\subset\mathcal{P}$, where
		\begin{eqnarray}
		\label{eqn:power_law_family_formula}
		p_\theta(x) = Z(\theta)^{-1} \big[q(x)^{\alpha-1} + \theta^Tf(x)\big]^{\frac{1}{\alpha-1}}\quad\text{for }~x\in\mathcal{X},
		\end{eqnarray}
		and $Z(\theta)$ is the normalizing constant (c.f. \cite{kumar2015minimization-2}).
	\end{definition}
	
	\begin{definition}
		The {\em $\alpha$-exponential family}, $\mathscr{E}_\alpha := \mathscr{E}_\alpha(q, f, \Theta)$, characterized by a $q\in\mathcal{P}$ and $k$ real valued functions $f_i, i = 1,\dots, k$ on $\mathcal{X}$, and parameter space $\Theta\subset\mathbb{R}^k$, is given by $\mathscr{E}_\alpha = \{ p_\theta : \theta\in\Theta\}\subset\mathcal{P}$,
		where
		\begin{eqnarray}
		\label{eqn:E_alpha_family_formula}
		p_\theta(x) = Z(\theta)^{-1} \big[q(x)^{1-\alpha} + \theta^Tf(x)\big]^{\frac{1}{1-\alpha}}\quad\text{for }~x\in\mathcal{X},
		\end{eqnarray}
		and $Z(\theta)$ is the normalizing constant (c.f. \cite{201609TIT_KumSas}).
	\end{definition}

Observe that the $\alpha$-power-law family was motivated from the minimization problem of relative $\alpha$-entropy $\mathscr{I}_\alpha$ subject to linear constraints on the underlying probability distribution \cite{kumar2015minimization-1} which, in turn, is inspired by the R\'enyi (or Tsallis) maximum entropy principle (in extensive statistical physics \cite[Eq.~(11)]{1998xxPhyA_Tsa}). Similarly, the $\alpha$-exponential family arose from the minimization problem of R\'enyi divergence $D_\alpha$ subject to constraints on the $\alpha$-escort of the underlying probability distribution \cite{201609TIT_KumSas} (motivated by non-extensive statistical physics \cite[Eq.~(22)]{1998xxPhyA_Tsa}). Since $\mathscr{I}_\alpha(p,q) = D_{{1}/{\alpha}}(p^{(\alpha)},q^{(\alpha)})$, one expects a close relationship between the two families. Indeed, the two families are closely related by the transformation $p\mapsto p^{(\alpha)}$ as the following lemma shows.

\begin{lemma}
\label{lem:connection_power_law_families}
	If the statistical manifold $S$ is an $\mathbb{M}^{(\alpha)}$ family characterized by a $q\in\mathcal{P}$, $f_i, i = 1,\dots, k$ on $\mathcal{X}$, and $\Theta$ then $S^{(\alpha)}$ is an $\mathscr{E}_{1/\alpha}$ family characterized by $q^{(\alpha)}$, ${f_i}/{\|q\|^{\alpha-1}}, i = 1,\dots, k$ on $\mathcal{X}$, and $\Theta$. 
\end{lemma}
\begin{proof}
Given an $\alpha$-power-law distribution
\begin{align}
p_{\theta}(x)=Z(\theta)^{-1}\Big[q(x)^{\alpha-1} + \sum\limits_{i=1}^k\theta_if_i(x)\Big]^{\frac{1}{\alpha-1}},
\end{align}
we have
\begin{align}
p_{\theta}(x)^\alpha &= Z(\theta)^{-\alpha}\Big[q(x)^{\alpha-1}+\sum\limits_{i=1}^k\theta_if_i(x)\Big]^{\frac{\alpha}{\alpha-1}}\nonumber\\
&= Z(\theta)^{-\alpha}\Big[q(x)^{\alpha-1}+\sum\limits_{i=1}^k\theta_if_i(x)\Big]^{\frac{1}{1-\frac{1}{\alpha}}}.
\end{align}
Therefore, after normalizing with $||p_\theta||^{\alpha}=\sum_x p_\theta(x)^\alpha$, we get the escort distribution
\begin{align}
\label{eq:escort1}
p_{\theta}^{(\alpha)}(x) &= \frac{Z(\theta)^{-\alpha}}{||p_\theta||^{\alpha}}\Big[q(x)^{\alpha-1}+\sum\limits_{i=1}^k\theta_if_i(x)\Big]^{\frac{1}{1-\frac{1}{\alpha}}}.
\end{align}

From (\ref{eq:escort1}),
\begin{align}
p_{\theta}^{(\alpha)}(x) &=\frac{Z(\theta)^{-\alpha}}{||p_\theta||^{\alpha}} \frac{1}{||q||^{-\alpha}||q||^{\alpha} } \Big[ q(x)^{\alpha-1} + \sum\limits_{i=1}^k\theta_i f_i(x)\Big]^{\frac{1}{1-\frac{1}{\alpha}}}\nonumber\\
&=\frac{Z(\theta)^{-\alpha}}{||p_\theta||^{\alpha}} \frac{1}{||q||^{-\alpha}} \left[\left( \frac{1}{||q||^{\alpha} }\right)^{{1-\frac{1}{\alpha}}} \left(q(x)^{\alpha-1} + \sum\limits_{i=1}^k\theta_i f_i(x) \right) \right]^{\frac{1}{1-\frac{1}{\alpha}}}\nonumber\\
&=\frac{Z(\theta)^{-\alpha}}{||p_\theta||^{\alpha} ||q||^{-\alpha}} \left[ \frac{1}{||q||^{\alpha-1} } \left(q(x)^{\alpha-1} + \sum\limits_{i=1}^k\theta_i f_i(x) \right) \right]^{\frac{1}{1-\frac{1}{\alpha}}}\nonumber\\
&=\frac{Z(\theta)^{-\alpha}}{||p_\theta||^{\alpha} ||q||^{-\alpha}} \left[  \frac{q(x)^{\alpha-1}}{||q||^{\alpha-1} } + \sum\limits_{i=1}^k  \theta_i\frac{ f_i(x)}{||q||^{\alpha-1} } \right]^{\frac{1}{1-\frac{1}{\alpha}}}\nonumber\\
&=\frac{Z(\theta)^{-\alpha}}{||p_\theta||^{\alpha} ||q||^{-\alpha}} \left[ \left( \frac{q(x)^{\alpha}}{||q||^{\alpha} }\right)^{\frac{\alpha-1}{\alpha}} + \sum\limits_{i=1}^k  \theta_i\frac{ f_i(x)}{||q||^{\alpha-1} } \right]^{\frac{1}{1-\frac{1}{\alpha}}}\nonumber\\
&=\frac{Z(\theta)^{-\alpha}}{||p_\theta||^{\alpha} ||q||^{-\alpha}}  \Bigg[ {q^{(\alpha)}(x)}^{1-\frac{1}{\alpha}} + \sum\limits_{i=1}^k \theta_i \frac{f_i(x)}{||q||^{\alpha-1} } \Bigg]^{\frac{1}{1-\frac{1}{\alpha}}}.\label{eq:escort3}
\end{align}

Let
\begin{align}
M(\theta) &\triangleq \frac{||p_\theta||^{\alpha} ||q||^{-\alpha}}{Z(\theta)^{-\alpha}}\\
h_i(x) &\triangleq \frac{f_i(x)}{||q||^{\alpha-1} }.
\end{align}
Then, rewriting the escort distribution yields
\begin{align}
\label{eq:escort4}
p_{\theta}^{(\alpha)}(x) = M(\theta)^{-1} \Big[ {q^{(\alpha)}(x)}^{1-\frac{1}{\alpha}} +\sum\limits_{i=1}^k \theta_i h_i(x) \Big]^{\frac{1}{1-\frac{1}{\alpha}}}.
\end{align}
This completes the proof. 
\begin{flushright}$\blacksquare$\end{flushright}
\end{proof}
	\vspace{-0.3cm}

Having established Lemma \ref{lem:connection_power_law_families}, one would expect that the two families are dual to each other with respect to the relative $\alpha$-entropic geometry. However, the answer is not in affirmative as proved in the following theorem. This example may be considered in the spirit of the duality mentioned in \cite[Sec. 3.3]{2000xxMIG_AmaNag}. 
\begin{theorem}
\label{thm:escort}
The ${1}/{\alpha}$-exponential family is not the dual of the $\mathbb{M}^{(\alpha)}$-family with respect to the metric induced by relative $\alpha$-entropy.
\end{theorem}
\begin{proof}
See Appendix A.
\begin{flushright}$\blacksquare$\end{flushright}
\end{proof}

\section{The General Framework and Applications}
	\label{sec:general_framework}
We apply the result in (\ref{eq:analogous_cramer_rao_inequality}) to a more general class of $f$-divergences. Observe from (\ref{eqn:alphadiv2}) that relative $\alpha$-entropy is a monotone function of an $f$-divergence not of the actual distributions but their $\alpha$-escort distributions.
Motivated by this, we first define a more general $f$-divergence and then show that these diveregnces also lead to generalized  Cram\'er-Rao lower bounds analogous to (\ref{eq:analogous_cramer_rao_inequality}). Although these divergences can be defined for positive measures, we restrict to probability measures here.
	
	\begin{definition}
	\label{defn:gen_f-divergence}
	Let $f$ be a strictly convex, twice continuously differentiable real valued function defined on $[0,\infty)$ with $f(1) = 0$ and $f''(1)\neq 0$. Let $F$ be a function that maps a probability distribution $p$ to another probability distribution $F(p)$. Then the {\em generalized $f$-divergence} between two probability distributions $p$ and $q$ is defined by
	\begin{equation}
	    \label{eqn:gen_f-divergence}
	    D_f^{(F)}(p,q) = \frac{1}{f''(1)}\cdot\sum_x F(q(x))f\left(\frac{F(p(x))}{F(q(x))}\right).
	\end{equation}
	\end{definition}
	
	Since $f$ is convex, by Jensen's inequality,
	\begin{eqnarray*}
	 D_f^{(F)}(p,q) & \ge & \frac{1}{f''(1)} \cdot f\left(\sum_x F(q(x))\cdot \frac{F(p(x))}{F(q(x))}\right)\\
	 & = & \frac{1}{f''(1)} \cdot f\left(\sum_x F(p(x))\right)\\
	 & = & \frac{1}{f''(1)}\cdot f(1)\\
	 & = & 0.
	\end{eqnarray*}
	Notice that, when $F(p(x)) = p(x)$, $D_f^{(F)}$ becomes the usual Csisz\'ar divergence. We now apply Eguchi's theory (see section \ref{sec:geometry_of_alpha_relative_entropy}) to $D_f^{(F)}$. The Riemannian metric on $S$ is specified by the matrix $G^{(f,F)}(\theta) = [g_{i,j}^{(f,F)}(\theta)]$, where
	\begin{eqnarray}
	\lefteqn{g_{i,j}^{(f,F)}(\theta) ~ := ~ g_{i,j}^{(D_f^{(F)})}(\theta)} \nonumber\\
	& = & -\frac{\partial}{\partial\theta_j'}\frac{\partial}{\partial\theta_i}D_f^{(F)}(p_{\theta},p_{\theta'})\bigg|_{\theta' = \theta} \nonumber\\
    & = & - \frac{\partial}{\partial\theta_j'}\frac{\partial}{\partial\theta_i} \sum_x F(p_{\theta'}(x))f\left(\frac{F(p_{\theta}(x))}{F(p_{\theta'}(x))}\right)\bigg|_{\theta' = \theta}\cdot \frac{1}{f''(1)}\nonumber\\
	& = & - \frac{\partial}{\partial\theta_j'}\left[\sum_x F(p_{\theta'}(x))f'\left(\frac{F(p_{\theta}(x))}{F(p_{\theta'}(x))}\right)\frac{F'(p_{\theta}(x))}{F(p_{\theta'}(x))}\partial_i p_{\theta}(x)\right]_{\theta' = \theta}\cdot \frac{1}{f''(1)}\nonumber\\
	& = & \left[\sum_x F'(p_{\theta}(x)) f''\left(\frac{F(p_{\theta}(x))}{F(p_{\theta'}(x))}\right)\frac{F(p_{\theta}(x))}{F(p_{\theta'}(x))^2}F'(p_{\theta'}(x))\partial_i p_{\theta}(x)\partial_j p_{\theta}(x)\right]_{\theta' = \theta}\cdot \frac{1}{f''(1)}\nonumber\\
	& = & \sum_x F(p_{\theta}(x))\cdot \partial_i \log F(p_{\theta}(x))\cdot \partial_j \log F(p _{\theta}(x))\nonumber\\
		\label{eqn:gen_metric}
	& = & E_{\theta^{(F)}}[\partial_i \log F(p_{\theta}(X))\cdot \partial_j \log F(p _{\theta}(X))],
	\end{eqnarray}
	 where $\theta^{(F)}$ stands for expectation with respect to the escort measure $F(p_\theta)$.
	 
	\begin{remark}
	Although the \textit{generalized} Csisz\'ar $f$-divergence is also a  Csisz\'ar$f$-divergence, it is not between $p$ and $q$. Rather, it is between the distributions $F(p)$ and $F(q)$. As a consequence, the metric induced by $D_f^{(F)}$ is different from the Fisher information metric, whereas the metric arising from all Csisz\'ar $f$-divergences is the Fisher information metric \cite{2010xxBPAS_AmaCic}.
	\end{remark}
	
	The following theorem extends the result in Theorem \ref{thm:alpha_CRLB} to a more general framework.
	
	\begin{theorem}[Generalized version of Cram\'{e}r-Rao inequality]
	\label{thm:gen_crlb}
	Let $\hat{\theta} = (\hat{\theta}_1,\dots,\hat{\theta}_m)$ be an unbiased estimator of $\theta = (\theta_1,\dots,\theta_m)$ for the statistical model $S$. Then there exists an unbiased estimator $\hat{\theta}^{F}$ of $\theta$ for the model $S^{(F)} = \{F(p) : p\in S\}$ such that $\text{Var}_{\theta^{(F)}}[\hat{\theta}^{(F)}(X)] - [G^{(f,F)}]^{-1}$ is positive semi-definite. Further, if $S$ is such that its escort model $S^{(F)}$ is exponential, then there exists efficient estimators for the escort model.
	\end{theorem}
	
	\begin{proof}
	Following the same steps as in Theorems \ref{thm:variance_and_norm_of_differential}-\ref{thm:alpha_CRLB} and Corollary \ref{cor:variance_and_norm_of_differential_inequality} produces \begin{eqnarray}
	\label{eq:analogous_cramer_rao_inequality2}
	c\text{Var}_{\theta^{(F)}}[\widehat{\theta}^{(F)}]c^t \ge c[G^{(f,F)}]^{-1}c^t
	\end{eqnarray}
	 for an unbiased estimator $\widehat{\theta}^{(F)}$ of $\theta$ for $S^{(F)}$. This proves the first assertion of the theorem. Now let us suppose that $p_\theta$ is model such that
	\begin{equation}
	\label{eqn:escort_model}
	    \log F(p_\theta(x)) = c(x) + \sum_{i=1}^k \theta_i h_i(x) - \psi(\theta).
	\end{equation}
	Then
	\begin{equation}
	\label{eqn:derivative_escort_model}
	    \partial_i \log F(p _{\theta}(x)) = h_i(x) - \psi(\theta).
	\end{equation}
	Let
    \[	
	\widehat{\eta}(x) := h_i(x) \quad \text{and} \quad \eta := E_{\theta^{(F)}}[\widehat{\eta}(X)].
	\]
	Since $E_{\theta^{(F)}}[\partial_i \log F(p _{\theta}(X))] = 0$, we have
	\begin{equation}
	    \label{eqn:derivative_of_potential}
	    \partial_i \psi(\theta) = \eta_i.
	\end{equation}
	Hence
	\begin{equation}
	\label{eqn:gen_metric_dual_expression1}
	    g_{i,j}^{(f,F)}(\theta) = E_{\theta^{(F)}}[(\widehat{\eta}_i(X) - \eta_i)(\widehat{\eta}_j(X) - \eta_j)]
	\end{equation}
	Moreover, since
	\begin{eqnarray}
	\label{eqn:second_derivative}
	\partial_i \partial_j\log F(p _{\theta}(x)) & = & \partial_i\left[\frac{1}{F(p _{\theta}(x))}\partial_j \log F(p _{\theta}(x))\right]\nonumber\\
	& = & \frac{1}{F(p _{\theta}(x))}\partial_i \partial_j F(p _{\theta}(x)) - \frac{1}{F(p _{\theta}(x))^2}\partial_i F(p _{\theta}(x))\partial_j F(p _{\theta}(x))\nonumber\\
    & = & \frac{1}{F(p _{\theta}(x))}\partial_i \partial_j F(p _{\theta}(x)) - \partial_i\log F(p _{\theta}(x))\partial_j\log F(p _{\theta}(x)),
	\end{eqnarray}
from (\ref{eqn:gen_metric}), we have
\begin{equation}
	\label{eqn:gen_metric_dual_expression2}
	    g_{i,j}^{(f,F)}(\theta) = - E_{\theta^{(F)}}[\partial_i \partial_j\log F(p _{\theta}(X))].
	\end{equation}	
	Hence, from (\ref{eqn:derivative_escort_model}) and (\ref{eqn:derivative_of_potential}), we have
	\begin{equation}
	    \label{eqn:dual_equation}
	    \partial_i \eta_j = g_{i,j}^{(f,F)}(\theta).
	\end{equation}
	This implies that $\eta$ is dual to $\theta$. Hence the generalized Fisher information matrix of $\eta$ is equal to the inverse of the generalized Fisher information matrix of $\theta$. Thus from
	(\ref{eqn:gen_metric_dual_expression1}), $\widehat{\eta}$ is an efficient estimator of $\eta$ for the escort model. This further helps us to find efficient estimators for $\theta$ for the escort model. This completes the proof. \begin{flushright}$\blacksquare$\end{flushright}
	\end{proof}
	
	Theorem~\ref{thm:gen_crlb} generalizes the dually flat structure of exponential and linear families with respect to the Fisher metric identified by Amari and Nagoaka
	\cite[Sec.~3.5]{2000xxMIG_AmaNag} to other distributions (as specified by \eqref{eqn:escort_model}) and a more widely applicable metric (as in Definition~\ref{defn:gen_f-divergence}).

	\section{Discussion}
	\label{sec:disc}
	Here we discuss some of the earlier works that have some commonalities with our work.
	\begin{enumerate}
	    \item Jan Naudts suggests an alternative generalization of the usual Cram\'er-Rao inequality in the context of Tsallis' thermostatistics \cite[Eq.~(2.5)]{2004xxJIPAM_Jan}.
	Their inequality is closely analogous to ours. It enables us to find a bound for the variance of an estimator of the underlying model (\textit{with respect to the escort model}) in terms of a generalized Fisher information ($g_{k l}(\theta)$) involving both the underlying ($p_\theta$) and its escort families ($P_\theta$).
	Their Fisher information, when the escort is taken to be $P_\theta = p_\theta^{(\alpha)}$, is given by
	\begin{align*}
	g_{k,l}(\theta) = \sum_x \frac{1}{p_\theta^{(\alpha)}(x)} \partial_k p_\theta(x)\partial_l p_\theta(x).
	\end{align*}
	The same in our case is
	\begin{align*}
	g_{k,l}^{(\alpha)}(\theta) = \sum_x \frac{1}{p_\theta^{(\alpha)}(x)} \partial_k p_\theta^{(\alpha)}(x)\partial_l p_\theta^{(\alpha)}(x).
	\end{align*}
	Also, $\partial_i p_\theta^{(\alpha)}$ and $\partial_i p_\theta$ are related by 
	\begin{align*}
	\partial_i p_\theta^{(\alpha)}(x) = \partial_i\left(\frac{p_\theta(x)^\alpha}{\sum_y p_\theta(y)^\alpha}\right) = \alpha\left[\frac{{p_{\theta}^{(\alpha)}(x)}}{p_{\theta}(x)}\partial_i p_{\theta}(x) - p_{\theta}^{(\alpha)}(x) \sum_y \frac{{p_{\theta}^{(\alpha)}(y)}}{p_{\theta}(y)}\partial_i p_{\theta}(y)\right].
	\end{align*}
	Moreover, while theirs bounds the variance of an estimator of the \textit{true distribution} with respect to the escort distribution, ours bounds the variance of an estimator of the \textit{escort distribution itself}. Their result is precisely the following.
	
	\vspace{0.2cm}
\noindent
\textbf{Theorem 2.1 of Jan Naudts \cite{2004xxJIPAM_Jan}} \textit{Let be given two families of pdfs $\left(p_{\theta}\right)_{\theta \in D}$ and $\left(P_{\theta}\right)_{\theta \in D}$ and corresponding expectations $E_{\theta}$ and $F_{\theta} .$ Let c be an estimator of $\left(p_{\theta}\right)_{\theta \in D},$ with scale function $F$. Assume that the regularity condition
\begin{align*}
F_{\theta} \frac{1}{P_{\theta}(x)} \frac{\partial}{\partial \theta^{k}} p_{\theta}(x)=0,
\end{align*}
holds. Let $g_{k l}(\theta)$ be the information matrix introduced before. Then, for all u and v in $\mathbb{R}^{n}$ is
\begin{align*}
\frac{u^{k} u^{l}\left[F_{\theta} c_{k} c_{l}-\left(F_{\theta} c_{k}\right)\left(F_{\theta} c_{l}\right)\right]}{\left[u^k v^l\frac{\partial^2}{ \partial\theta^l\partial\theta^k} F(\theta)\right]^{2}} \geq \frac{1}{v^{k} v^{l} g_{k l}(\theta)}.
\end{align*}
}
\vspace{0.2cm}

\item Furuichi \cite{furuichi2009} defines a generalized Fisher information based on the $q$-logarithmic function and gives a bound for the variance of an estimator with respect to the escort distribution. Given a random variable $X$ with the probability density function $f(x)$, they define the $q$-score function $s_{q}(x)$ based on the $q$-logarithmic function and $q$-Fisher information $J_{q}(X) = E_{q}\left[s_{q}(X)^{2}\right]$,
where $E_{q}$ stands for expectation with respect to the escort distribution $f^{(q)}$ of $f$ as in \eqref{eqn:escort_distribution}. Observe that
\begin{align}
\label{eqn:furuchi_fisher}
    J_{q}(X) & = E_{q}\left[s_{q}(X)^{2}\right]\nonumber\\
    & = E_{q}\left[f(X)^{2-2q}\left(\frac{d}{dX}\log f(X)\right)^2\right],
\end{align}
whereas our Fisher information in this setup, following \eqref{eqn:g-alpha-expansion}, is
\begin{equation}
\label{eqn:q-fisher_information}
    g^{(q)}(X) = E_{q}\left[ \left(\frac{d}{dX}\log f(X)\right)^2\right] - \left(E_{q}\left[ \frac{d}{dX}\log f(X)\right]\right)^2,
\end{equation}
Interestingly, they also bound the variance of an estimator of the escort model with respect to the escort model itself as in our case. Their main result is the following.
\vspace{0.2cm}

\noindent
\textbf{Theorem 3.3 of Furuichi \cite{furuichi2009}}: \textit{Given the random variable $X$ with the probability density function $p(x)$, the $q$-expectation value $\mu_{q} = E_{q}[X]$, and the $q$-variance $\sigma_{q}^{2} = E_{q}\left[\left(X-\mu_{q}\right)^{2}\right]$, we have a $q$-Cramér-Rao inequality
\begin{align*}
J_{q}(X) \geq \frac{1}{\sigma_{q}^{2}}\left(\frac{2}{\int p(x)^{q} d x}-1\right) \quad \text { for } q \in[0,1) \cup(1,3).
\end{align*}
Immediately, we have
\begin{align*}
J_{q}(X) \geq \frac{1}{\sigma_{q}^{2}} \quad \text { for } q \in(1,3).
\end{align*}
}
\vspace{0.2cm}

\item Lutwak et al. \cite{200501TIT_LutYanZha} derives a Cram\'er-Rao inequality in connection with extending Stam's inequality for the generalized Gaussian densities. Their inequality finds lower bound for the $p$-th moment of the given density ($\sigma_p[f]$) in terms of a generalized Fisher information. Their Fisher information $\phi_{p, \lambda}[f]$, when specialised to $p=q=2$, is given by
\[
\phi_{2, \lambda}[f] = \left\{E\Big[f(X)^{2\lambda - 2}\Big(\frac{d}{dX}\log f(X)\Big)^2\Big]\right\}^{\frac{1}{2}},
\]
which is closely related to that of Furuichi's \eqref{eqn:furuchi_fisher} upto a change of measure $f\mapsto f^{(\lambda)}$, which, in turn, related to ours \eqref{eqn:q-fisher_information}. Moreover, while they use relative $\alpha$-entropy to derive their moment-entropy inequality, they do not do so while defining their Fisher information and hence obtain a different Cram\'er-Rao inequality. Their result is reproduced as follows.
\vspace{0.2cm}

\noindent
\textbf{Theorem 5 of Lutwak et al. \cite{200501TIT_LutYanZha}}: \textit{Let $p \in[1, \infty], \lambda \in(1 /(1+p), \infty),$ and $f$ be a density. If $p<\infty,$ then $f$ is assumed to be absolutely continuous; if $p=\infty,$ then $f^{\lambda}$ is assumed to have bounded variation. If $\sigma_p[f], \phi_{p, \lambda}[f]<\infty,$ then
\begin{align*}
\sigma_p[f]\phi_{p, \lambda}[f] \geq \sigma_p[G]\phi_{p, \lambda}[G],
\end{align*}
where $G$ is the generalized Gaussian density.
}
\vspace{0.2cm}

\item Bercher \cite{bercher2012generalized} derived a two parameter extension of Fisher information and a generalized Cram\'{e}r-Rao inequality which bounds the $\alpha$ moment of an estimator. Their Fisher information, when specialised to $\alpha = \beta = 2$, reduces to
\begin{align*}
I_{2, q}[f ; \theta] = E_q\left[\frac{f^{(q)}(X ; \theta)}{f(x ; \theta)} \left(\frac{\partial}{\partial \theta} \log f^{(q)}(X ; \theta)\right)^2\right],
\end{align*}
where $E_q$ stands for expectation with respect to the escort distribution $f^{(q)}$. Whereas, following \eqref{eqn:RiemannianOnS-alpha}, our Fisher information in this setup is
\begin{eqnarray*}
    g^{(q)}(\theta) = \frac{1}{q^2} E_{q}\left[\left(\frac{\partial}{\partial \theta} \log f^{(q)}(X ; \theta)\right)^2\right].
\end{eqnarray*}
Thus our Fisher information differs from his by the factor ${f(x;\theta)}/{q^2 f^{(q)}(x,\theta)}$ inside the expectation. Note that $q$ in their result is analogous to $\alpha$ in our work. The main result of Bercher \cite{bercher2012generalized} is reproduced verbatim as follows.
\vspace{0.2cm}

\noindent
\textbf{Theorem 1 of Bercher \cite{bercher2012generalized}}: \textit{Let $f(x; \theta)$ be a univariate probability density function defined over a subset $X$ of $\mathbb{R},$ and $\theta \in \Theta$ a parameter of the density. Assume that $f(x ; \theta)$ is a jointly measurable function of $x$ and $\theta$, is integrable with respect to $x$, is absolutely continuous with respect to $\theta,$ and that the derivative with respect to $\theta$ is locally integrable. Assume also that
$q>0$ and that $M_{q}[f ; \theta]$ is finite. For any estimator $\hat{\theta}(x)$ of $\theta,$ we have
\begin{align*}
E\left[|\hat{\theta}(x)-\theta|^{\alpha}\right]^{\frac{1}{\alpha}} I_{\beta, q}[f ; \theta]^{\frac{1}{\beta}} \geq\left|1+\frac{\partial}{\partial \theta} E_{q}[\hat{\theta}(x)-\theta]\right|,
\end{align*}
with $\alpha$ and $\beta$ H\"{o}lder conjugates of each other, i.e., $\alpha^{-1}+\beta^{-1}=1, \alpha \geq 1,$ and where the quantity
\begin{align*}
I_{\beta, q}[f ; \theta] = E\left[\left|\frac{f(x ; \theta)^{q-1}}{M_{q}[f ; \theta]} \frac{\partial}{\partial \theta} \ln \left(\frac{f(x ; \theta)^{q}}{M_{q}[f ; \theta]}\right)\right|^{\beta}\right],
\end{align*}
where $M_{q}[f ; \theta] := \int f(x ; \theta)^{q}~dx$, is the generalized Fisher information of order $(\beta, q)$ on the parameter $\theta .$ 
}
\vspace{0.2cm}

\item In our recent work \cite{mishra2020generalized}, we introduced the definition of Bayesian relative $\alpha$-entropy. As before, we define $S = \{p_\theta: \theta = (\theta_1,\dots,\theta_k)\in\Theta\}$ as a $k$-dimensional sub-manifold of $\mathcal{P}$ and
\begin{align}
\label{eqn:denormalized_manifold}
\tilde{S} := \{\tilde{p}_{\theta}(x) = p_{\theta}(x)\lambda(\theta) : p_{\theta}\in S\},
\end{align}
where $\lambda$ is a probability distribution on $\Theta$. Then, $\tilde{S}$ is a sub-manifold of $\tilde{\mathcal{P}}$. Let $\tilde{p}_\theta, \tilde{p}_{\theta'}\in \tilde{S}$. We defined relative $\alpha$-entropy of $\tilde{p}_\theta$ with respect to $\tilde{p}_{\theta'}$ as
\begin{align*}
 \mathcal{I}_{\alpha}(\tilde{p}_{\theta},\tilde{p}_{\theta'}) := &\frac{\lambda(\theta)}{1-\alpha}\log\sum_x p_\theta(x) (\lambda(\theta')p_{\theta'}(x))^{\alpha-1}+ \lambda(\theta')\nonumber\\
& - \lambda(\theta)\left[\frac{\log\sum_xp_\theta(x)^\alpha}{\alpha (1-\alpha)} - \{1 + \log \lambda(\theta)\} -\frac{1}{\alpha} \log\sum_xp_{\theta'}(x)^\alpha\right],
\end{align*}
which is a generalization of \eqref{eqn:alphadiv}. In \cite{mishra2020generalized}, we used this divergence to derive the Bayesian version of Theorem~\ref{thm:gen_crlb} proved in the previous section. The result in \cite{mishra2020generalized}, stated in the following theorem, reduces to deterministic $\alpha$-CRLB in the absence of prior, Bayesian CRLB when $\alpha \rightarrow 1$, and deterministic CRLB when $\alpha \rightarrow 1$ in the absence of prior.
\vspace{0.2cm}

\noindent
\textbf{Bayesian $\alpha$-Cram\'{e}r-Rao inequality \cite{mishra2020generalized}}:
	\label{thm:Bayesian_alpha_CRLB}
	\textit{Let $S = \{p_{\theta} : \theta = (\theta_1,\dots,\theta_m)\in\Theta\}$ be the given statistical model and let $\tilde{S}$ as in \eqref{eqn:denormalized_manifold}. Let $\hat{\theta} = (\hat{\theta}_1,\dots,\hat{\theta}_m)$ be an unbiased estimator of $\theta = (\theta_1,\dots,\theta_m)$ for the statistical model $S$. Then
	\begin{equation*}
	\label{eqn:Bayesian_alpha_cramerrao}
	   \int \text{Var}_{\theta^{(\alpha)}}\left[\frac{\tilde{p_\theta}(X)}{p_\theta^{(\alpha)}(X)}(\hat{\theta}(X) - \theta)\right] d\theta
	    \ge \left\{E_\lambda\big[G_\lambda^{(\alpha)}\big]\right\}^{-1},
	\end{equation*}
	where $\theta^{(\alpha)}$ denotes expectation with respect to $p_{\theta}^{(\alpha)}$.}
	
\end{enumerate}

	\section{Summary}
	\label{sec:summary}
	We studied the geometry of probability distributions with respect to a generalized version of the Csisz\'ar $f$-divergences under the Eguchi's framework. Amari and Nagoaka \cite{2000xxMIG_AmaNag} established the Cram\'er-Rao inequality from the information geometry of the usual relative entropy $\mathscr{I}(p,q)$. Following their procedure, we formulated an analogous inequality from the generalized Csisz\'ar $f$-divergences. This result, in fact, leads the usual Cram\'er-Rao inequality to its escort $F(p)$ by the transformation $p\mapsto F(p)$. However, this reduction is not coincidental here because the Riemannian metric derived from all Csisz\'ar $f$-divergences is the Fisher information metric and the divergence studied here is a Csisz\'ar $f$-divergence, not between $p$ and $q$, but between $F(p)$ and $F(q)$. 
	
	Nonetheless, the generalized version of the Cram\'er-Rao inequality enables us to find unbiased and efficient estimators for the escort of the underlying model. This theory when specified to relative $\alpha$-entropy, gives rise to an $\alpha$-CRLB for the $\alpha$-escort of the underlying distribution. This further elucidates the relation between the two important power-law families namely, $\alpha$-power-law and $\alpha$-exponential. Indeed, as a consequence of this general theory, we proved that neither of these families is a dual of the other with respect to the usual Fisher information metric or its $\alpha$-version. Such counterexamples are not available from the Amari-Nagaoka derivation of the CRLB.
	
	\appendix
	\section{Proof of Theorem \ref{thm:escort}}
	\label{app:proof}
	
	Taking log on both sides of (\ref{eq:escort4}),
\begin{eqnarray}
\label{eq:escort5}
\log p_{\theta}^{(\alpha)}(x) &= -\log M(\theta) + {\frac{1}{1-\frac{1}{\alpha}}}\log\Big[ {q^{(\alpha)}(x)}^{1-\frac{1}{\alpha}} + \sum\limits_{i=1}^k \theta_i h_i(x) \Big].
\end{eqnarray}
Partial derivative produces
\begin{eqnarray}
\partial_i \log p_{\theta}^{(\alpha)}(x) &= -\partial_i\log M(\theta) + {\frac{1}{1-\frac{1}{\alpha}}}\frac{h_i(x)}{ \Big[ {q^{(\alpha)}(x)}^{1-\frac{1}{\alpha}} +\sum\limits_{i=1}^k \theta_i h_i(x) \Big]},
\end{eqnarray}
or
\begin{eqnarray}
\label{eq:escort6}
\partial_i \log p_{\theta}^{(\alpha)}(x) &= -\partial_i\log M(\theta) + \frac{\alpha}{\alpha-1} \frac{f_i(x)}{ \Big[q(x)^{\alpha-1} + \sum\limits_{i=1}^k\theta_i f_i(x) \Big]}.
\end{eqnarray}
Taking expectation on both sides of (\ref{eq:escort6}), we obtain
\begin{eqnarray}
\label{eq:escort7}
E_{\theta^{(\alpha)}} \left[ \partial_i \log p_{\theta}^{(\alpha)}(x)\right] &= -\partial_i\log M(\theta) + \frac{\alpha}{\alpha-1} E_{\theta^{(\alpha)}} \Bigg[\frac{f_i(X)}{  q(X)^{\alpha-1} + \sum\limits_{i=1}^k\theta_i f_i(X)}\Bigg].
\end{eqnarray}
Since the expected value of the score function vanishes (left hand side of (\ref{eq:escort7})), we have
\begin{equation}
    \label{eq:escort8}
\partial_i\log M(\theta) = \frac{\alpha}{\alpha-1} E_{\theta^{(\alpha)}} \vast[\frac{f_i(X)}{  q(X)^{\alpha-1} + \sum\limits_{i=1}^k\theta_i f_i(X)}\vast].
\end{equation}

Substituting (\ref{eq:escort8}) into (\ref{eq:escort6}), we get
\begin{eqnarray}
\label{eq:escort9}
\partial_i \log p_{\theta}^{(\alpha)}(x) 
&=& \frac{\alpha}{\alpha-1}  \frac{f_i(x)}{ \Big[  q(x)^{\alpha-1} + \sum\limits_{i=1}^k\theta_i f_i(x) \Big]} - \frac{\alpha}{\alpha-1}   E_{\theta^{(\alpha)}} \vast[\frac{f_i(X)}{ q(X)^{\alpha-1} + \sum\limits_{i=1}^k\theta_i f_i(X) }\vast] \nonumber\\
& =: & \widehat{\eta}_i(x) - \eta_i,
\end{eqnarray}
where
\[
\widehat{\eta}_i(x) := \frac{\alpha }{\alpha-1}\frac{ f_i(x)}{ \Big[q(x)^{\alpha-1} + \sum\limits_{i=1}^k\theta_i f_i(x) \Big]} \text{ and } \eta_i := E_{\theta^{(\alpha)}}[\widehat{\eta}_i(X)]
\]
Moreover, (\ref{eq:escort8}) implies that $\log M(\theta)$ should be the potential (if exists).

The Riemmanian metric becomes
\begin{align}
\label{eqn:cramer-rao-equality-dual}
g_{ij}^{(\alpha)}(\theta) = \frac{1}{\alpha^2}E_{\theta^{(\alpha)}}[(\widehat{\eta}_i(X) - \eta_i)(\widehat{\eta}_j(X) - \eta_j)].
\end{align}
This further strengthens our expectation that the $\eta_i$'s are dual parameters to $\theta_i$'s. However, it is surprising that it is not so as we shall see now. We have
\begin{eqnarray}
\eta_j & = & \frac{\alpha }{\alpha-1}E_{\theta^{(\alpha)}}\vast[ \frac{f_j(X)}{  q(X)^{\alpha-1} + \sum\limits_{j=1}^k\theta_j f_j(X) }\vast]\nonumber\\
& = & \frac{\alpha }{\alpha-1}\sum\limits_x  \vast[\frac{f_j(x)}{  q(x)^{\alpha-1} + \sum\limits_{j=1}^k\theta_j f_j(x)}\vast] p_{\theta}^{(\alpha)}(x).
\end{eqnarray}
Let $R_{\theta}(x) = q(x)^{\alpha-1} + \sum\limits_{j=1}^k\theta_j f_j(x)$. Partial differentiation produces
\begin{eqnarray}
\label{eq:derivative_eta}
\frac{\partial \eta_j}{\partial \theta_i} 
&=& \frac{\alpha }{\alpha-1}\frac{\partial}{\partial \theta_i}\left(\sum_x \frac{f_j(x)p_{\theta}^{(\alpha)}(x)}{R_{\theta}(x)}\right)\nonumber\\
&=& \frac{\alpha }{\alpha-1}\sum_x  \frac{R_{\theta}(x)f_j(x) \partial_i p_{\theta}^{(\alpha)}(x) - f_j(x) p_{\theta}^{(\alpha)}(x) \partial_i(R_{\theta}(x))}{(R_{\theta}(x))^2}\nonumber\\
&=& \frac{\alpha }{\alpha-1}\sum_x \left[ \frac{f_j(x)}{R_{\theta }(x)}\partial_i p_{\theta}^{(\alpha)}(x) -  p_{\theta}^{(\alpha)}(x)\frac{f_j(x)}{R_{\theta}(x)} \frac{f_i(x)}{R_{\theta}(x)} \right].
\end{eqnarray}

From (\ref{eq:escort6}) - (\ref{eq:escort9}), we have 
\begin{eqnarray}
\label{eq:expr1}
\frac{\alpha}{\alpha-1}\frac{f_i(x)}{ R_{\theta }} &= \partial_i \log p_{\theta}^{(\alpha)}(x) + \eta_i.
\end{eqnarray}

Substituting (\ref{eq:expr1}) into (\ref{eq:derivative_eta}) gives
\begin{align}
\frac{\partial \eta_j}{\partial \theta_i} 
&= \frac{\alpha }{\alpha-1}\sum_x \bigg[ \partial_i p_{\theta}^{(\alpha)}(x) \frac{\alpha-1}{\alpha} (\partial_i \log p_{\theta}^{(\alpha)}(x) + \eta_i) \nonumber\\
&-  p_{\theta}^{(\alpha)}(x) \left(\frac{\alpha-1}{\alpha}\right)^2  (\partial_j \log p_{\theta}^{(\alpha)}(x) + \eta_j) (\partial_i \log p_{\theta}^{(\alpha)}(x) + \eta_i) \bigg]\nonumber\\
&= \sum_x \bigg[\partial_i p_{\theta}^{(\alpha)}(x) (\partial_i \log p_{\theta}^{(\alpha)}(x) - \eta_i)\nonumber\\
&-  p_{\theta}^{(\alpha)}(x) \left(\frac{\alpha-1}{\alpha}\right)  (\partial_j \log p_{\theta}^{(\alpha)}(x) + \eta_j) (\partial_i \log p_{\theta}^{(\alpha)}(x) + \eta_i)\bigg]\nonumber\\
&= \sum_x \bigg[\partial_i p_{\theta}^{(\alpha)}(x) \partial_i \log p_{\theta}^{(\alpha)}(x) +  \eta_i \partial_i p_{\theta}^{(\alpha)}(x)\nonumber\\
&-  \left(\frac{\alpha-1}{\alpha}\right) p_{\theta}^{(\alpha)}(x)  \bigg(\partial_j \log p_{\theta}^{(\alpha)}(x) \partial_i \log p_{\theta}^{(\alpha)}(x) + \eta_i \partial_i \log p_{\theta}^{(\alpha)}(x) \nonumber\\
&  +\eta_i \partial_j \log p_{\theta}^{(\alpha)}(x) + \eta_i \eta_j 
\bigg)\bigg]\nonumber
\end{align}
\begin{align}
&= \sum_x \bigg[\partial_i p_{\theta}^{(\alpha)}(x) \partial_i \log p_{\theta}^{(\alpha)}(x) + \cancelto{=0}{   \eta_i \partial_i p_{\theta}^{(\alpha)}(x)}\nonumber\\
&-  \left(\frac{\alpha-1}{\alpha}\right)\bigg(p_{\theta}^{(\alpha)}(x)\partial_j \log p_{\theta}^{(\alpha)}(x) \partial_i \log p_{\theta}^{(\alpha)}(x) + \cancelto{=0}{ \eta_j p_{\theta}^{(\alpha)}(x) \partial_i \log p_{\theta}^{(\alpha)}(x)} \nonumber\\
&+\cancelto{=0}{\eta_i p_{\theta}^{(\alpha)}(x)\partial_j \log p_{\theta}^{(\alpha)}(x)} +  \eta_i \eta_j p_{\theta}^{(\alpha)}(x) \bigg)\bigg]\nonumber\\
&= \sum_x \bigg[\partial_i p_{\theta}^{(\alpha)}(x) \partial_i \log p_{\theta}^{(\alpha)}(x) - p_{\theta}^{(\alpha)}(x) \frac{1}{p_{\theta}^{(\alpha)}(x)} \partial_j p_{\theta}^{(\alpha)}(x) \partial_i \log p_{\theta}^{(\alpha)}(x)  \nonumber\\ &+\frac{1}{\alpha}\bigg(p_{\theta}^{(\alpha)}(x)\partial_j \log p_{\theta}^{(\alpha)}(x) \partial_i \log p_{\theta}^{(\alpha)}(x)\bigg) - \left(\frac{\alpha-1}{\alpha}\right) \eta_i \eta_j p_{\theta}^{(\alpha)}(x) \bigg) \bigg]\nonumber\\
&= \sum_x \bigg[\frac{1}{\alpha}\bigg(p_{\theta}^{(\alpha)}(x)\partial_j \log p_{\theta}^{(\alpha)}(x) \partial_i \log p_{\theta}^{(\alpha)}(x)\bigg) - \left(\frac{\alpha-1}{\alpha}\right) \eta_i \eta_j p_{\theta}^{(\alpha)}(x) \bigg) \bigg]\nonumber\\
&= \alpha g_{ij}^{(\alpha)}(\theta) - \left(\frac{\alpha-1}{\alpha}\right) \eta_i \eta_j.\label{eq:derivative_eta1}
\end{align}

This shows that $\eta_i$ cannot be the dual parameters of $\theta_i$ for the statistical model $\mathbb{M}^{(\alpha)}$. This completes the proof. 

\begin{acknowledgements}
The authors are indebted to Prof. Rajesh Sundaresan of the Indian Institute of Science, Bengaluru for his helpful suggestions and discussions that improved the presentation of this material substantially.
\end{acknowledgements}

On behalf of all authors, the corresponding author states that there is no conflict of interest.
\bibliographystyle{spmpsci}      
\bibliography{main}   

\end{document}